\documentclass[a4paper,12pt]{article}
\usepackage[margin=1in]{geometry}

\usepackage{comment}
\usepackage{amssymb,amsmath}
\usepackage{amsthm}
\usepackage{xspace}
\usepackage{enumerate}
\usepackage[pdfborder={0 0 0}]{hyperref}

\newtheorem{theorem}{Theorem}
\newtheorem{lemma}[theorem]{Lemma}
\newtheorem{proposition}[theorem]{Proposition}
\newtheorem{claim}[theorem]{Claim}
\newtheorem{hypothesis}[theorem]{Hypothesis}
\newtheorem{mylemma}[theorem]{Lemma}

\theoremstyle{definition}
\newtheorem{mydefinition}[theorem]{Definition}
\newtheorem{remark}[theorem]{Remark}
\newtheorem*{remark*}{Remark}

\newcommand{\newextmathcommand}[2]{%
    \newcommand{#1}{\ensuremath{#2}\xspace}
}

\newcommand{\renewextmathcommand}[2]{%
    \renewcommand{#1}{\ensuremath{#2}\xspace}
}

\newextmathcommand{\DyckTwoReach}{\mathsf{D_2Reach}}
\newextmathcommand{\Lang}{\mathcal L}
\newextmathcommand{\RLang}{\mathcal L'}
\newextmathcommand{\A}{\mathcal A}
\newextmathcommand{\Ak}{\mathcal A_k}
\renewextmathcommand{\AA}{\mathcal A'}
\newextmathcommand{\B}{\mathcal B}
\newextmathcommand{\Rp}{\mathcal R}
\newextmathcommand{\PDA}{\mathcal P}
\newextmathcommand{\PDAA}{\mathcal P'}
\newextmathcommand{\SAT}{\mathsf{SAT}}
\newextmathcommand{\SharpSAT}{\mathsf{\# SAT}}
\newextmathcommand{\TAUT}{\mathsf{TAUT}}
\newextmathcommand{\SETH}{\mathsf{SETH}}
\newextmathcommand{\NSETH}{\mathsf{NSETH}}
\newextmathcommand{\NTIME}{\mathsf{NTIME}}
\newextmathcommand{\CONTIME}{\mathsf{coNTIME}}
\newextmathcommand{\DTIME}{\mathsf{DTIME}}
\newextmathcommand{\MATIME}{\mathsf{MATIME}}
\newextmathcommand{\coMATIME}{\mathsf{co\text{-}MATIME}}
\renewextmathcommand{\P}{\mathsf{P}}
\newextmathcommand{\NP}{\mathsf{NP}}
\newextmathcommand{\MA}{\mathsf{MA}}
\newextmathcommand{\PSPACE}{\mathsf{PSPACE}}
\newextmathcommand{\HardestA}{\A_0}
\newextmathcommand{\AlmostHardestA}{\A'_0}
\newextmathcommand{\prestar}{\mathrm{Pre}^*(R)}

\newextmathcommand{\eps}{\varepsilon}
\renewcommand{\epsilon}{\varepsilon}
\newcommand{\sset}{\subseteq}

\newextmathcommand{\partto}{\rightharpoonup}

\newextmathcommand{\opr}{\text{\textup{\texttt{(}}}}
\newextmathcommand{\clr}{\text{\textup{\texttt{)}}}}
\newextmathcommand{\ops}{\text{\textup{\texttt{[}}}}
\newextmathcommand{\cls}{\text{\textup{\texttt{]}}}}
\newextmathcommand{\BrTwo}{\{ \opr, \clr, \ops, \cls \}}
\newextmathcommand{\OpenBr}{\{ \opr, \ops \}}
\newextmathcommand{\CloseBr}{\{ \clr, \cls \}}

\newextmathcommand{\qfinal}{q_\textsc{f}}
\newcommand{\set}[1]{\{ #1 \}}
\newextmathcommand{\lmarker}{{\lhd}}
\newextmathcommand{\rmarker}{{\rhd}}
\newextmathcommand{\llmarker}{{\trianglelefteq}}
\newextmathcommand{\rrmarker}{{\trianglerighteq}}
\newextmathcommand{\shuf}{\parallel}
\newcommand{\negate}{\overline}
\newcommand{\twonpda}{\mathcal{A}}
\newextmathcommand{\pureSigma}{\Sigma \setminus \{\lmarker, \rmarker\}}

\newextmathcommand{\sdelim}{{\mathtt{;}}}
\newextmathcommand{\strack}{\diamond}
\newextmathcommand{\Bin}{\{0, 1\}}

\newextmathcommand{\aAND}{\textsc{and}}
\newextmathcommand{\aOR}{\textsc{or}}
\newextmathcommand{\hash}{\mathtt\#}
\newextmathcommand{\amdelima}{\star}
\newextmathcommand{\amdelimb}{*}

\newextmathcommand{\offset}{\mathsf{offset}}
\newextmathcommand{\ind}{\mathsf{index}}
\newextmathcommand{\vmark}{\#}
\newextmathcommand{\edgesep}{*}
\newextmathcommand{\poly}{\mathrm{poly}}
\newextmathcommand{\polylog}{\mathrm{polylog}}

\newextmathcommand{\AllNonTerminals}{\overrightarrow{V^2}}
\newextmathcommand{\NonTerminals}{\mathsf{NT}}
\newcommand{\NT}{\NonTerminals}
\newcommand{\nont}[2]{\overrightarrow{#1 #2\vphantom{d}}}

\newcommand*{\br}{\nobreak\discretionary{}{\hbox{}}{}}

\newenvironment{claimproof}{\begin{proof}}{\end{proof}}

\title{Subcubic Certificates for CFL Reachability}

\author{Dmitry Chistikov${}^1$ \and Rupak Majumdar${}^2$ \and Philipp Schepper${}^3$}
\date{%
\normalsize
$^1$%
Centre for Discrete Mathematics and its Applications (DIMAP) \&
Department of Computer Science, University of Warwick,
Coventry, United Kingdom\\
\texttt{d.chistikov@warwick.ac.uk}\\
\mbox{}\\
$^2$%
Max Planck Institute for Software Systems, Kaiserslautern, Germany\\
\texttt{rupak@mpi-sws.org}\\
\mbox{}\\
$^3$
CISPA Helmholtz Center for Information Security, Saarbr\"ucken, Germany \&
Saarbr\"ucken Graduate School of Computer Science, Saarland Informatics Campus, Germany\\
\texttt{philipp.schepper@cispa.saarland}}

\begin{document}

\maketitle

\begin{abstract}
Many problems in interprocedural program analysis can be modeled as the context-free language (CFL) reachability problem on graphs and can be
solved in cubic time.
Despite years of efforts, there are no known truly sub-cubic algorithms for this problem.
We study the related \emph{certification} task:
given an instance of CFL reachability, are there small and efficiently
checkable certificates for the existence and for the non-existence of a path?
We show that, in both scenarios,
there exist succinct certificates
($O(n^2)$ in the size of the problem)
 and these certificates can be checked in subcubic (matrix multiplication) time.
The certificates are based on grammar-based compression of paths (for positive instances) and on invariants represented as
matrix constraints (for negative instances).
Thus, CFL reachability lies in nondeterministic and co-nondeterministic \emph{subcubic} time.

A natural question is whether faster algorithms for CFL reachability
will lead to faster algorithms for combinatorial problems such as Boolean satisfiability (SAT).
As a consequence of our certification results,
we show that there cannot be a fine-grained reduction from SAT to CFL reachability
for a conditional lower bound stronger than $n^\omega$, 
unless the nondeterministic strong exponential time hypothesis (NSETH) fails.

Our results extend to related
subcubic equivalent problems: pushdown reachability and
two-way nondeterministic pushdown automata (2NPDA) language recognition.
For example, we describe
succinct certificates for pushdown non-reachability (inductive invariants)
and observe that they can be checked in matrix multiplication time.
We also extract a new hardest 2NPDA language, capturing the
``hard core'' of all these problems. 
\end{abstract}

\newpage

\section{Introduction}
\label{sec:intro}

Context-free reachability is a fundamental problem in interprocedural program analysis, verification of recursive
programs, and database theory \cite{DolevEK82,Yannakakis90,MelskiReps,RHS95,BEM97}.
For a fixed context-free language (CFL) $\mathcal{L}$ over an alphabet $\Sigma$,
given a directed graph $G = (V,E)$, an edge-labeling function $\lambda : E \rightarrow \Sigma$, 
and two vertices $s, t \in V$, the $\mathcal{L}$-reachability problem asks if there is a path from $s$ to
$t$ in $G$ such that the word formed by concatenating the labels along the path belongs to $\mathcal{L}$.
It is well-known that the problem can be solved in time cubic in the size of the graph for any fixed CFL.
However, despite many years of efforts, we only know speedups by logarithmic factors (i.e., to $O(n^3 / \log n)$)
\cite{Rytter85,Chaudhuri},
leading to a conjecture that no better algorithms are possible for this and several related problems 
\cite{HeintzeMcAllester}.
In recent years, a number of results in fine-grained complexity give credence to the conjecture
by demonstrating various conditional lower bounds for the problem \cite{ChatterjeeCP18,AbboudBW15a,MathiasenP21},
but even so the possibility of algorithms with running time~$n^\omega$ or above has not been ruled out.
Here, $\omega < 2.4$ is the matrix multiplication exponent~\cite{CoppersmithW90,Williams12}.

In this paper, we study the problem of \emph{certifying} an instance of CFL reachability.
Intuitively, this problem asks for easily verifiable proofs of inclusion or non-inclusion. 
Given a (positive or negative) instance of CFL reachability, we ask if there is an efficiently checkable
proof that will convince anyone that the instance is indeed positive or negative. 

Formally, a \emph{certificate system} for CFL reachability consists of two algorithms (the checkers),
one for positive instances and one for negative instance.
Each checker takes as input an instance of the problem and an additional string (called the certificate) and accepts
or rejects.
The positive (resp.\ negative) checker is \emph{complete} if for each positive (resp.\ negative) instance,
there is a certificate that makes it accept, and \emph{sound} if for each negative (resp.\ positive) instance,
there is no certificate that makes it accept.
Of course, since the instance can be decided in cubic time, a certificate system is non-trivial only if the
checkers run in subcubic time (in the size of the instance).

Our main result shows the existence of \textbf{subcubic certificate systems for CFL reachability}: every positive
or negative instance has a \emph{quadratic} certificate and a checker that runs in $O(n^\omega)$ time.
\begin{itemize}
\item
For a positive instance of the problem, a naive certificate is a path from $s$ to~$t$ witnessing inclusion.
Unfortunately, this is not an efficient certificate, since it is known that the shortest path can be exponentially
long in the size of the graph.
We show that the shortest path is well-compressible by a context free grammar of size ${O}(n^2)$ in the
number of vertices of the graph.
Moreover, given such a compressed representation, there is a checker verifying in time $O(n^2)$ that
the grammar indeed encodes a witness path.
\item
For a negative instance of the problem, a certificate is an inductive invariant that demonstrates \emph{non-}reachability.
We show that such an inductive invariant can be represented as relations between a constant number of $n\times n$ matrices, and there is a checker
verifying in time $O(n^\omega)$ that such an encoding does represent an inductive invariant.
Additionally, if we allow randomization, there is a randomized checker running in $O(n^2)$ time.
\end{itemize}
Summing up, CFL reachability can be certified in subcubic time.
In retrospect, the certificate system is simple but illuminates a conceptually new aspect of an old problem. 
Certificate systems make it possible to separate two possibly independent phases of computation,
\emph{finding} a solution to a computational problem and \emph{verifying} it.

We consider complexity-theoretic implications.
Impagliazzo and Paturi \cite{IP01} introduced the strong exponential time hypothesis ($\SETH$), which informally states that \SAT has 
no algorithms better than exhaustive search.
Over the years, $\SETH$ has become a fundamental assumption relative to which many fine-grained complexity results are proved~\cite{WilliamsSurvey}.
For example, $\SETH$ implies current (quadratic) algorithms for orthogonal vectors or edit distance problems are optimal.
A natural question is if $\SETH$ also implies that cubic algorithms for CFL reachability are optimal.

Our result shows that such a reduction would be very difficult to find. 
Carmosino et al.~\cite{Carmosino} extended $\SETH$ to the nondeterministic strong exponential time hypothesis ($\NSETH$), which
states that there is no algorithm for Boolean tautology better than exhaustive search,
even with nondeterministic guessing.
They show that both proving and refuting $\NSETH$ imply breakthroughs in computational complexity.
Our subcubic certification result implies that any conditional lower bound for CFL reachability from \SAT and $\SETH$
will show that $\NSETH$ does not hold.

A model checking problem closely related to CFL reachability is
pushdown reachability~\cite{BEM97,FinkelWW97}.
Our results lead to \textbf{a subcubic certificate system for pushdown reachability} too,
by extracting quadratic certificates from the standard saturation-based algorithm
and the triplet construction for PDA to CFG conversion.
Indeed, by exploiting fine-grained reductions between CFL reachability, pushdown reachability, the emptiness problem for pushdown automata,
and the recognition problem for two-way nondeterministic pushdown automata (2NPDA), 
we show all these problems (as well as other related problems known in the literature) have subcubic certificate systems.
Our constructions and reductions have several implications.
First, succinct certificates for pushdown (non-)reachability
checkable in subcubic time is a new observation; it can have 
potentially practical application in \emph{checking} proofs of programs \cite{Necula97} 
and in ``exports'' of model checking
such as certificate set analysis in trust management systems~\cite{JhaR04}.
Second, our reductions lead to a new insight beyond certification.
We identify
\textbf{a new hardest 2NPDA language}, that is, a fixed
2NPDA language $L_0$ such that for every 2NPDA language $L$
there is a homomorphism $h$ such that $w \in L$ iff $h(w) \in L_0$. 
A different hardest language was previously found by Rytter~\cite{Rytter-hardest} using language-theoretic techniques.
However, our proof and reductions strengthen the link between
2NPDA language recognition and CFL reachability,
pointing to the hardest instances of the latter.

\paragraph*{Related work.}

In a quest to classify the
complexity of problems in P,
\emph{fine-grained reductions} interlink the asymptotic running time
of algorithms for various problems.
A fine-grained reduction shows that
a faster algorithm for one problem automatically
implies a faster algorithm for another problem.
Conversely, the existence of fine-grained reductions can be 
interpreted as \emph{conditional lower bounds:}
no faster algorithm exists, unless a state-of-the-art algorithm
for a well-known problem is actually suboptimal.
For example, a truly sub-quadratic algorithm for Orthogonal Vectors 
will lead to a $2^{(1-\epsilon) n}$-time algorithm for \SAT, breaking $\SETH$ \cite{WilliamsSurvey}.
Similarly, the $k$-Clique conjecture states that no (randomized or deterministic) algorithm can detect a $k$-Clique
on an $n$-vertex graph in time $O(n^{\frac{\omega k}{3} - \epsilon})$ for $\epsilon > 0$.
Abboud et al.~\cite{AbboudBW15a} show a reduction from the $k$-Clique problem
to CFL \emph{recognition}, giving a conditional lower bound of order $n^\omega$
and matching Valiant's $\tilde{O}(n^\omega)$ upper bound for the problem~\cite{Valiant75}.
This lower bound applies to CFL reachability as well.
Chatterjee et al.~\cite{ChatterjeeCP18}, using Lee's result~\cite{Lee02}, reduce Boolean matrix multiplication
to Dyck-$k$ reachability (for growing~$k$), showing that faster algorithms for the latter
avoiding matrix multiplication would be a breakthrough.
Chatterjee and Osang~\cite{ChatterjeeO17} show a similar reduction
to PDA emptiness.

More broadly, a range of problems in formal languages are now being approached
with tools from modern algorithms and complexity~\cite{Fernau19}.
Our work contributes to this ongoing effort.
Backurs and Indyk~\cite{BackursI16} and
Bringmann, Gr{\o}nlund, and Larsen~\cite{BringmannGL17} find
\SAT-based conditional lower bounds for regular expression matching problems.
Oliveira and Wehar~\cite{OliveiraW18} show
reductions between triangle finding, 3SUM,
and the non-emptiness of intersection of two or three DFA.
Potechin and Shallit~\cite{PotechinS20} show a reduction from Orthogonal Vectors to
the acceptance problem for (a subclass of) NFA and a reduction from triangle finding to (unary) NFA acceptance.
Fernau and Krebs~\cite{FernauK17} establish conditional lower bounds
for a variety of automata-theoretic problems beyond~P.
Wehar and co-authors have shown that faster algorithms for
various intersection non-emptiness problems have consequences
for structural complexity classes~\cite{Wehar14,SwernofskyW15,OliveiraW20}.
We discuss further related work in Section~\ref{s:fl}.

\section{Context-free reachability and Dyck-2 reachability}
\label{s:pre}

Let $\mathcal{L}$ be a fixed language.
Given a directed graph $G = (V, E)$, an edge-labeling function $\lambda \colon E \rightarrow \Sigma$,
and two vertices $s, t \in V$, 
the $\mathcal{L}$-reachability problem asks 
if there is a path from $s$ to $t$
(possibly repeating vertices and edges)
such that the word formed by concatenating the labels along the path
belongs to $\mathcal{L}$ \cite{Yannakakis90}.
When $\mathcal{L}$ is a fixed context-free language, the problem is called \emph{CFL reachability.}
CFL reachability plays a foundational
role in several areas within computer science.
To the best of our knowledge, it first appears
in the work by Dolev, Even, and Karp~\cite{DolevEK82}
as the combinatorial core
in the security analysis of a cryptographic protocol.
Yannakakis~\cite{Yannakakis90} and
Melski and Reps~\cite{MelskiReps} elucidate the
role of this problem
in the context of database theory
and interprocedural program analysis, respectively,
providing in particular a historical sketch.

In formal language theory, $\mathcal{L}$-reachability and CFL reachability
can be seen as providing an algorithmic perspective on
the classic definition of rational index
of a language~\cite{BoassonCN81,P92}.
These problems have also been studied under the name ``regular realizability''
(see, e.g.,~\cite{Vyalyi11,VyalyiR15}).

For CFL reachability, without loss of generality, the fixed language can be assumed to be 
the \emph{Dyck-$2$ language}.
This is the language of balanced parentheses with two kinds of parenthesis symbols.
Formally, it is the context free language over the alphabet $\BrTwo$ defined by the
following context-free grammar:
\[
S \to SS \mid \opr\; S\; \clr \mid \ops\; S\; \cls \mid \epsilon
\]
The \emph{Dyck-$2$ reachability} problem, denoted \DyckTwoReach,
is the $\mathcal{L}$-reachability problem when $\mathcal{L}$ is the Dyck-$2$ language. 
%
\begin{claim}
\label{claim:pda-reach-to-dyck-2-reach}
Let $(G, \lambda, s, t)$ be an instance of the  CFL reachability problem.
There is a linear-time reduction (in the bit-size of the input)
to an instance $(G', \lambda', s', t')$ of the Dyck-$2$ reachability problem.
\end{claim}

We call an algorithm \emph{truly subcubic}
if it has (worst-case) running time $O(n^{3-\eps})$
for some constant $\eps > 0$,
where $n$ denotes the bit length of the input.
Practical implementations use a summarization-based $O(|V|^3)$ algorithm \cite{RHS95};
note that $|V| \le n$.
Using Rytter's trick \cite{Rytter85}, Chaudhuri~\cite{Chaudhuri} shows that 
the $\mathcal{L}$-reachability problem is $O(|V|^3/\log |V|)$ for any fixed context-free language.
However, no truly subcubic algorithm is known for this problem.
The best known conditional lower bound for the problem is has order~$|V|^\omega$.
On the other hand, Dyck-$1$ reachability (the language of balanced parentheses with one kind of parentheses)
can be solved in time $\tilde{O}(|V|^\omega)$~\cite{Bradford2018,BringmannPersonalCommunication,MathiasenP21},
matching best conditional lower bounds.

\section{Certificates for reachability and non-reachability}
\label{s:cert}

In this section we show that, while truly subcubic algorithms for
Dyck-2 reachability are not known,
solutions to Dyck-2 reachability have small and efficiently checkable certificates.

An instance $(G, \lambda, s, t)$ of \DyckTwoReach
is a yes-instance if there is a walk from $s$ to $t$ labeled with
a string from Dyck-2, and a no-instance otherwise.

\begin{mydefinition}
We say that \DyckTwoReach has \emph{subcubic
  certificates} for yes-instances (respectively, no-instances)
if, for some real number~$\eps > 0$, there is an algorithm $M$ and a
function $p(x) = O(x^{3-\epsilon})$ such
that for every instance $(G, \lambda, s, t)$ of \DyckTwoReach:
\begin{description}
\item[(completeness)]
if the instance is a yes-instance (respectively, no-instance), then
there is a string $u$ of length $p(|V|)$, called a \emph{certificate},
such that $M$ accepts $(G,\lambda,
s,t, u)$ in 
$p(|V|)$ time,
and
\item[(soundness)]
if the instance is a no-instance (respectively, yes-instance), 
then for every string $u$ of length $p(|V|)$,
the algorithm $M$ rejects $(G, \lambda, s,t, u)$ in
$p(|V|)$ time.
\end{description}
\end{mydefinition}
(Note that the running time of $M$ is
subcubic in $|V|$, which is at most the bit size of the instance,
and \emph{not} in the size of the certificate.)
That is, a subcubic certificate for yes- and no-instances allows us to verify, given the additional certificate, 
whether an instance of \DyckTwoReach is a positive or a negative instance in sub-cubic time.

We will refer extensively to \emph{walks} in labelled directed graphs.
For a labelled directed graph $(V,E,\lambda\colon E\rightarrow \BrTwo)$, 
a walk from $u\in V$ to $v\in V$ is a sequence of edges
$\pi:= e_0\ldots e_{k}$ from $E$, for $k\geq 0$, such that
for each $i \in \{1, \ldots, k\}$, edge $e_{i-1}$ arrives
at the same vertex that edge $e_i$ departs from,
and moreover $e_0$ departs from~$u$ and $e_{k}$ arrives at~$v$.
This walk is \emph{valid} if the word
$\lambda(e_0) \ldots \lambda(e_k)$ belongs to the Dyck-$2$ language.
A~subwalk of a walk $e_0\ldots e_k$ is a contiguous
subsequence $e_i\ldots e_j$ of edges, possibly empty.

\subsection{Certificates for yes-instances: compressed walks}
\label{s:cert:reach}

We describe our certificate system for yes-instances of \DyckTwoReach.
These certificates are witnesses for reachability.
We fix an instance of \DyckTwoReach:
$G = (V, E)$ a directed graph,
$\lambda \colon E \rightarrow \BrTwo$ an edge-labeling function,
and $s, t \in V$ source and target vertices.


A first attempt is to provide a valid walk as a certificate (witness).
However, it is well-known that the shortest valid walk can be
exponential in the size of the input, namely it can be of length
$\exp\Theta(|V|^2 / \log |V|)$, and this bound
is tight~\cite{P92}.
(For an intuition, one can think of a pushdown automaton accepting
 only words of length exponential in its size and longer.)
The main observation to get subcubic certificates is that there is
always some valid walk (including the shortest one in particular)
that is well-compressible and that has a small representation ($O(|V|^2)$
in the size of the graph) and it is
efficient to check (in time $O(|V|^2)$)
that such a compressed walk is indeed a valid walk.
Moreover, for every no-instance, one cannot get any valid walks,
compressed or otherwise.

The following definition ``inlines'' the concept of a straight-line program,
which is an ``acyclic'' context-free grammar that generates one word only.
Straight-line programs are at the core of general-purpose compression algorithms
such as LZ77 (see, e.g.,~\cite{Lohrey12}).

\begin{mydefinition}
For an instance of \DyckTwoReach,
denote by \AllNonTerminals a fresh copy of the set~$V^2$,
written as $\AllNonTerminals = \{ \nont u v \mid (u, v) \in V^2 \}$.
A \emph{walk scheme} is
a context-free grammar with the set of terminal symbols~$E$,
a set of nonterminal symbols~$\NT \sset \AllNonTerminals$,
and
the axiom $\nont s t \in \NT$,
where:
\begin{itemize}
\item
for each nonterminal $\nont u v \in \NT$
there is exactly one production, which moreover has the form:
\begin{enumerate}[(a)]
\item
$
\nont u v \to \nont u w \  \nont w v
$
for some $w \in V$,
or
\item
$
\nont u v \to e \; \nont x y \; f
$
for some edges $e = (u,x) \in E$ and $f = (y, v) \in E$
with $\lambda(e) \cdot \lambda(f) \in \{ \opr \clr, \ops \cls \}$,
or
\item
$
\nont u u \to \eps
$
for some $u \in V$, and
\end{enumerate}
\item
the directed graph with vertices \NT and the following set of edges
is acyclic:
\begin{equation}
\label{eq:graph-of-slp}
\left \{ \,
         (\nont a b, \nont c d) \mid
         \text{$\nont c d$ occurs on the right-hand side of
               the production of $\nont a b$}
\, \right \}.
\end{equation}
\end{itemize}
\end{mydefinition}

\begin{proposition}
Every walk scheme has size $O(|V|^2)$ and bit size $O(|V|^2 \log |V|)$.
\end{proposition}

\begin{theorem}
\label{th:cert-yes}
The following statements hold:
\begin{itemize}
\item
An instance of \DyckTwoReach is a yes-instance
if and only if
there exists a walk scheme for it.
\item
There is a deterministic algorithm that runs in time $O(|V|^2)$
and decides if a given grammar is a walk scheme
for a given instance of \DyckTwoReach.
\end{itemize}
\end{theorem}

For the proof of Theorem~\ref{th:cert-yes},
we need the following auxiliary result.

\begin{mylemma}
\label{l:prune-alternatives}
Let $\mathcal G$ be a context-free grammar
with $L(\mathcal G) \ne \emptyset$.
Suppose $\mathcal G$ contains more than one production
with the same nonterminal on the left-hand side.
Then by removing all of them but one we can obtain a grammar
$\mathcal G'$ with $L(\mathcal G') \ne \emptyset$.
\end{mylemma}

\begin{proof}[Proof of Theorem~\ref{th:cert-yes}]
We split the proof into three parts.

\paragraph*{Soundness.}
We first suppose that for a given instance of \DyckTwoReach
there exists a walk scheme, $\mathcal W$,
and show that the instance must be a yes-instance.
Consider the directed graph from
the acyclicity condition in the definition of walk schemes,
denote it $D$.
We will consider all vertices of~$D$, i.e., nonterminals from \NT,
in any reversed topological ordering. In other words,
whenever
               $\nont c d$ occurs on the right-hand side of
               the production of $\nont a b$,
we will consider $\nont c d$ before $\nont a b$.
We will show by induction that, for every $\nont u v \in \NT$,
the (one) word generated by $\nont u v$ is a valid walk
from~$u$ to~$v$.
(Recall that a walk is valid if it is labelled by a Dyck-2 word.)
Indeed, it suffices to consider the three types of productions:
\begin{enumerate}[(a)]
\item
for a production of the form
$
\nont u v \to \nont u w \  \nont w v
$,
we know from the inductive hypothesis that $\nont u w$ generates
a valid walk from $u$ to $w$, and $\nont w v$ a valid walk from
$w$ to $v$, so their concatenation is a valid walk from $u$ to $v$;
\item
for a production of the form
$
\nont u v \to e \; \nont x y \; f
$
with edges $e = (u,x) \in E$ and $f = (y, v) \in E$,
we know from the inductive hypothesis that $\nont x y$ generates
a valid walk from $x$ to~$y$, and since
$\lambda(e) \cdot \lambda(f) \in \{ \opr \clr, \ops \cls \}$,
the result of the concatenation is a valid walk from $u$ to $v$;
\item
finally, productions of the form
$
\nont u u \to \eps
$
correspond to trivial valid walks (containing no edges) and represent
the induction base.
\end{enumerate}
As the axiom of the grammar~$\mathcal W$ is $\nont s t$, we conclude
that there is a valid walk from $s$ to~$t$, which means that
the instance of \DyckTwoReach we consider is a yes-instance.

\paragraph*{Completeness.}
In the converse direction, let us prove that
that every yes-instance of \DyckTwoReach
has a walk scheme. Consider such an instance, $(G, \lambda, s, t)$,
and consider a walk from~$s$ to~$t$, call it~$\pi$.
We construct a walk scheme in several steps.

First consider a context-free grammar $\mathcal G$
with the set of terminal symbols~$E$,
set of nonterminal symbols~$\AllNonTerminals$,
and
axiom $\nont s t$.
The set of productions is determined as follows.
For each nonterminal $\nont u v \in \AllNonTerminals$,
we include all productions of the form:
\begin{itemize}
\item
$
\nont u v \to \nont u w \  \nont w v
$
for all $w \in V$; 
\item
$
\nont u v \to e \; \nont x y \; f
$
where $e = (u,x) \in E$ and $f = (y, v) \in E$
such that $\lambda(e) \cdot \lambda(f) \in \{ \opr \clr, \ops \cls \}$;
\item
$\nont u u \to \eps$
for all $u \in V$.
\end{itemize}
Induction on the structure of $\pi$ shows that
$\pi \in L(\mathcal G)$, so
$L(\mathcal G) \ne \varnothing$.

We can now
prune the set of productions of the grammar $\mathcal G$
using Lemma~\ref{l:prune-alternatives}, as well as
apply standard procedures of removing useless (non-productive or
unreachable) nonterminals in context-free grammars
(see, e.g.,~\cite[Section~7.1]{HopcroftMotwaniUllman}).
We perform these steps until all three have no effect on the grammar.
The resulting grammar $\mathcal W$ satisfies all conditions
in the definition of walk schemes, except possibly the acyclicity condition.
We claim that $\mathcal W$ must satisfy that condition too.
Indeed, the transformations applied so far ensure that $L(\mathcal W) \ne \emptyset$.
Let $\NT \sset \AllNonTerminals$ be the set of nonterminals of $\mathcal W$.
Assume for the sake of contradiction that
the directed graph with vertices \NT and edges~\eqref{eq:graph-of-slp}
contains a directed cycle. Let $\nont a b \in \NT$ be a vertex on this cycle.
Since all nonterminals of $\mathcal W$ are reachable and productive,
there exists a valid parse tree with respect to $\mathcal W$ that
contains a node labelled by~$\nont a b$.
By definition of the graph, and since every nonterminal in~$\mathcal W$
has exactly one production, this node has a descendant labelled with~$\nont a b$.
By the same reasoning, this descendant also has a descendant labelled with~$\nont a b$,
etc., which cannot be the case as the tree is finite.
This contradiction means that the graph must be acyclic, so
$\mathcal W$ is in fact a walk scheme.

\paragraph*{Verification algorithm.}
The condition~$\NT \sset \AllNonTerminals$
and the choice of the axiom can be checked in time $O(|V|^2)$.
The fact that there is exactly one production per nonterminal
can be checked under the same time constraints;
and so can the form of these productions
and compatibility with the instance of \DyckTwoReach.
Finally, depth-first search--based topological sort procedure
can be used to detect the existence of directed cycles;
it runs in time linear in the number of edges, which is at most
$|V|^2$.
\end{proof}

\begin{remark}
There is nothing special about Dyck-$2$ in the construction, and a similar certificate
can be constructed for any fixed CFG.
\end{remark}

We already mentioned a link to compressed words above.
Our proof of Theorem~\ref{th:cert-yes} finds
a context-free grammar that generates exactly one
word and has $O(|V|^2)$ nonterminals in Chomsky normal form.
Importantly, while it is in general a \PSPACE-complete problem to decide
whether such a compressed word is accepted by a pushdown automaton
(see, e.g., the survey~\cite[section~9.4]{Lohrey12} and references therein),
our grammar has special structure, leading to
an efficient verification algorithm.

\subsection{Certificates for no-instances: inductive invariants}
\label{s:cert:non-reach}

Fix an instance of \DyckTwoReach. 
For ease of notation, we will assume that $V = \{1,\ldots,|V|\}$.
A certificate for no-instances will be a \emph{separator},
as defined next.
Such a certificate is essentially an inductive invariant, certifying non-reachability.

Let $A_{\opr}$, $A_{\ops}$, $A_{\clr}$, $A_{\cls}$ be four $0$--$1$ matrices
of size $|V| \times |V|$ that are adjacency matrices for the graph $G$ restricted
to sets of edges with labels $\opr$, $\ops$, $\clr$, $\cls$, respectively.

\newcommand{\tobool}[1]{\mathrm{bool}(#1)}
For a nonnegative integer matrix $N$, denote
by $\tobool N$ the matrix obtained from $N$ by replacing
every nonzero element by~$1$.
Let $I$ denote the $|V| \times |V|$ identity matrix.
We write $A \le B$ for matrices $A = (a_{i j})$ and $B = (b_{i j})$
of the same size whenever $a_{i j} \le b_{i j}$ for all $i$, $j$.

\begin{mydefinition}
A \emph{separator} for an instance of \DyckTwoReach
is a sextuple of $|V| \times |V|$ matrices,
$(M_{S}, M_{S S}, M_{\opr S}, M_{\ops S},\br M_{\opr S \clr}, M_{\ops S \cls})$,
where all entries belong to $\{0, 1, \ldots, |V|^2\}$,
and moreover all entries of $M_S$ belong to $\{0, 1\}$,
and such that the following ten conditions are satisfied:
\begin{equation}
\label{eq:sep}
\begin{aligned}
I &\le M_S,                                    &
A_{\opr} \cdot M_S &= M_{\opr S},  &
A_{\ops} \cdot M_S &= M_{\ops S},
\\
M_S \cdot M_S &= M_{S S},          &
M_{\opr S} \cdot A_{\clr} &= M_{\opr S \clr},  &
M_{\ops S} \cdot A_{\cls} &= M_{\ops S \cls},
\\
\tobool{M_{S S}} &\le M_S,          & \phantom{\text{and}}
\tobool{M_{\opr S \clr}} &\le M_S,  &
\tobool{M_{\ops S \cls}} &\le M_S,
\ \text{ and }
(M_S)_{s,t} = 0,
\end{aligned}
\end{equation}
where $s$ and $t$ are the source and target vertex
in the instance of \DyckTwoReach.
\end{mydefinition}

\begin{proposition}
Every separator has $O(|V|^2)$ entries and bit size $O(|V|^2 \log |V|)$.
\end{proposition}

\begin{theorem}
\label{th:cert-no}
The following statements hold:
\begin{itemize}
\item
An instance of \DyckTwoReach is a no-instance
if and only if
there exists a separator for it.
\item
There is a deterministic algorithm that runs in time $O(|V|^\omega)$
and decides if a given sextuple of $|V| \times |V|$ matrices is a separator
for a given instance of \DyckTwoReach.
\item
There is a randomized algorithm that runs in time $O(|V|^2)$
and decides if a given sextuple of $|V| \times |V|$ matrices is a separator
for a given instance of \DyckTwoReach.
In the case it is, the algorithm never errs;
otherwise the algorithm flags an issue with probability $\ge 0.5$.
\end{itemize}
\end{theorem}

\begin{proof}
We split the proof into four parts.

\paragraph*{Completeness.}
First consider a no-instance of \DyckTwoReach.
Take the matrix $M_S = (m_{i j})$, where each $m_{i j}$ is
$1$ if there is a valid walk from vertex~$i$ to vertex~$j$.
It is clear that $m_{s t} = 0$, because the instance is a no-instance.
We now show that picking the other matrices
$M_{S S}, M_{\opr S}, M_{\ops S}, M_{\opr S \clr}, M_{\ops S \cls}$
so that all the five matrix equalities among the constraints~\eqref{eq:sep}
are satisfied leads to the satisfaction of the remaining
(four) inequality constraints. Indeed:
\begin{itemize}
\item 
$I \le M_S$ because for each vertex~$i$ the empty walk from~$i$ to~$i$
is valid;
\item
$\tobool{M_{S S}} \le M_S$ because the concatenation of two valid walks
is a valid walk;
\item
$\tobool{M_{\opr S \clr}} \le M_S$ and
$\tobool{M_{\ops S \cls}} \le M_S$ because
every walk $e \cdot \pi \cdot e'$ is valid whenever
$\pi$ is valid and $e$ and $e'$ are labelled by a matching
pair of parentheses, either \opr, \clr or \ops, \cls.
\end{itemize}
This shows that there is a separator for each no-instance.

\paragraph*{Soundness.}
In the converse direction, consider an arbitrary instance of \DyckTwoReach.
We show that for every valid walk~$\pi$
from a vertex~$u$ to a vertex~$v$ in the graph, all separators
must satisfy the condition $m_{u v} = 1$ where $M_S = (m_{i j})$.
(It then follows that yes-instances have no
 separators.)
We use induction on the label of walk~$\pi$, which is simply the concatenation of
individual edge labels:
\begin{itemize}
\item
The base case is the empty label, $\eps$.
The walk~$\pi$ must then be the empty walk, from some vertex~$u$ to itself.
We recall that $I \le M_S$ for every separator;
so indeed $m_{i i}$ must be set to $1$ for all vertices~$i$,
and for the chosen vertex $i = u$ in particular.
\item
If the walk~$\pi$ is labelled by $\alpha \cdot \beta$,
where both $\alpha$ and $\beta$ are nonempty Dyck-2 words,
then there exists a vertex $w$ such that $\pi = \pi' \cdot \pi''$
and $\pi'$ and $\pi''$ are valid walks from $u$ to $w$ and from $w$ to $v$,
respectively.
By the inductive hypothesis, $m_{u,w} = m_{w,v} = 1$.
Since $\tobool{M_{S S}} = \tobool{M_S \cdot M_S} \le M_S$,
we conclude that $m_{u,v} = 1$ in this case as well.
\item
Finally, suppose the label of the walk~$\pi$ is $\opr \alpha \clr$,
for some Dyck-2 word~$\alpha$.
(The case $\ops \alpha \cls$ is analogous.)
Then $\pi = e \cdot \pi' \cdot f$, where $e$ and $f$ are individual
edges, say from $u$ to $u'$ and from $v'$ to $v$ (for some $u', v' \in V$),
and $\pi'$ is a valid walk from $u'$ to $v'$.
The edges $e = (u, u')$ and $f = (v', v)$ have labels $\opr$ and $\clr$, respectively.
By the inductive hypothesis, $m_{u' v'} = 1$.
We now observe that
$\tobool{M_{\opr S \clr}} =
 \tobool{M_{\opr S} \cdot A_{\clr}} =
 \tobool{A_{\opr} \cdot M_S \cdot A_{\clr}} \le M_S$.
On the left-hand side, the matrix product has a positive entry in position $u v$,
because $\left(A_{\opr}\right)_{u,u'} =
         \left(A_{\clr}\right)_{v',v} = 1$
by the definition of $A_{\opr}$ and $A_{\clr}$.
Therefore $m_{u v} = 1$.
\end{itemize}
This concludes the proof of the first assertion of the theorem.

\paragraph*{Deterministic algorithm.}
The algorithm from the second assertion of the theorem verifies
all conditions in the definition of separator directly.
This means in particular five matrix multiplications where the factors
are matrices with elements from $\{0, \ldots, |V|\}$ (worst-case time $O(|V|^\omega)$),
four inequalities between individual matrices (worst-case time $O(|V|^2)$), and
a single equality constraint on one of the entries (constant time).

\begin{remark*}
This algorithm reduces
the verification of separators to 5~matrix
multiplications over the nonnegative integers.
While this result has complexity-theoretic consequences
(see Section~\ref{s:sowhat} below),
it may appear unsatisfactory,
as many theoretical algorithms for fast matrix multiplication are impractical.
This brings the randomized algorithm to the~fore.
\end{remark*}

\paragraph*{Randomized algorithm.}
The algorithm from the final assertion of the theorem
is the same as the previous one, except that
instead of \emph{computing} matrix multiplication it runs Freivalds' algorithm
for \emph{verifying} matrix multiplication~\cite{Freivalds79}.

Recall that Freivalds' algorithm
for verifying $A \cdot B = C$ for some $n\times n$ matrices
$A$, $B$, and $C$ proceeds by picking a $0$--$1$ vector $u \in \set{0,1}^n$ uniformly
at random and checking if $A \cdot (B u) = C u$. The algorithm runs in $O(n^2)$ time and 
has error probability $1/2$.
The properties of the algorithm are transferred directly to give a $O(|V|^2)$ bound.
Since we have five products to check, we reduce the error probability
in an individual check to $1/16$ by running it $4$~times,
so that the overall error probability is at most $5/16 \le 1/2$.
\end{proof}

\begin{remark}
For the deterministic verification algorithm,
it suffices to specify the $0$--$1$ matrix $M_S$ only, because
the other five matrices can be computed in time $O(|V|^\omega)$ from~it.
\end{remark}

\begin{remark}
Once again,
there is nothing special about the Dyck-2 language in our certificate system.
One can readily see that the conditions we impose on separators
correspond to the following context-free grammar for the Dyck-2 language:
\begin{align*}
S &\to S S \ |\ P \clr\ |\ Q \cls\ |\ \eps &
P &\to \opr S &
Q &\to \ops S
\enspace.
\end{align*}
Replacing this grammar with a different one,
we obtain 
a certificate system (for no-instances) for the CFL reachability problem
where the fixed CFL is represented by any fixed CFG.
\end{remark}

\begin{remark}
In a model of computation with unit-cost integer arithmetic,
integer matrix multiplication can be verified
in deterministic time $O(n^2)$~\cite{KorecW14}.
For RAM with
$O(\log n)$-bit arithmetic operations,
derandomization of Freivalds' algorithm
is an open problem even in the nondeterministic setting.
However, if the number of errors in the product is guaranteed to be $O(n^{2-\epsilon})$,
then a deterministic $O(n^{3-\epsilon})$-time algorithm is known~\cite{Kunnemann18}.
\end{remark}

\section{Complexity implications}
\label{s:sowhat}

\paragraph*{Complexity-theoretic summary of Section~\ref{s:cert}.}

%
Leaving out sharper bounds on certificate size and
$\polylog(n)$ factors
(required in the Turing model),
Theorems~\ref{th:cert-yes} and~\ref{th:cert-no} imply:

\begin{theorem}
\label{th:cert}
$
\DyckTwoReach
\in
         \NTIME   (n^2     ) \cap
       \CONTIME   (n^\omega) \cap
      \coMATIME_1 (n^2     )
$.
\end{theorem}

For this summary, we recall (cf.~\cite{Tell19})
that $L \in \MATIME_1(t)$
(Merlin-Arthur time, introduced by Babai~\cite{Babai85})
iff there exists a deterministic machine
$M$ that takes inputs
$x,y,z$ where $|y| = |z| = O(t(|x|))$, runs in time $O(t(|x|))$, and
such that for every $x$,
\begin{align*}
x \in L \; \Rightarrow \;  \exists y. \Pr_z[M(x,y,z) \mbox{ accepts}] = 1,\quad\quad
x \not\in L \; \Rightarrow \;  \forall y. \Pr_z[M(x,y,z) \mbox{ accepts}] \leq 1/2,
\end{align*}
where the probability is with respect to the uniform distribution
of $z$ in $\{0,1\}^{t(|x|)}$.
Finally, $\coMATIME_1(t)$ is
the class of \emph{complements} of languages in $\MATIME_1(t)$.

\paragraph*{Fine-grained complexity of \DyckTwoReach.}
Fine-grained complexity research shows that even small improvements
in (the exponent of) the running time of many algorithmic problems,
such as orthogonal vectors or edit distance, would automatically give
faster algorithms for Boolean satisfiability, \SAT%
~\cite{WilliamsSurvey}.
Would improvements over Chaudhuri's $O(n^3 / \log n)$-time algorithm for \DyckTwoReach
also have consequences for \SAT?
Here we show
that subcubic certificates
give an answer to this question.


In fine-grained complexity,
perhaps the most influential hypothesis,
and the ultimate source of many lower bounds,
is
the \emph{strong exponential-time hypothesis} ($\SETH$)~\cite{IP01}, stating (roughly) that
there is no algorithm for \SAT better than exhaustive enumeration.
The \emph{non-deterministic strong exponential-time hypothesis} ($\NSETH$)~\cite{Carmosino}
extends it further.

\begin{hypothesis}[$\SETH$]
For every $\varepsilon > 0$, there exists a $k$ so that $k$-$\SAT$ 
is not in $\DTIME[2^{n(1 - \varepsilon)}]$, where $k$-$\SAT$ is the language of all satisfiable Boolean
formulas in $k$-CNF. 
\end{hypothesis}

\begin{hypothesis}[$\NSETH$]
For every $\varepsilon > 0$, there exists a $k$ so that $k$-$\TAUT$ 
is not in $\NTIME[2^{n(1 - \varepsilon)}]$, where $k$-$\TAUT$ is the language of all Boolean tautologies
in $k$-DNF. 
\end{hypothesis}

In both hypotheses, $n$ is the number of variables.
It is unknown whether \SETH and \NSETH are true.
\NSETH implies \SETH, and \SETH implies $\P \ne \NP$.
Carmosino et al.~\cite{Carmosino} explore consequences of \NSETH and show that both proving
and refuting it would lead to interesting consequences.
In particular, \NSETH implies the \emph{absence} of fine-grained reductions from \SAT to a number of problems
and $\lnot \NSETH$ implies circuit lower bounds.

It turns out that,
because of our subcubic certificate systems (Section~\ref{s:cert}),
there exists no fine-grained reduction from \SAT (as well as from any $\SETH$-hard problem)
to \DyckTwoReach that would imply hardness beyond $n^\omega$, unless $\NSETH$ fails.

Because of space constraints, we relegate the formal definition
of fine-grained reductions
to Appendix~\ref{app:fgc}.
Intuitively,
a fine-grained reduction from $(L, t(n))$
to $(\DyckTwoReach, n^{c})$ 
means that, for every $\eps>0$, an $O(n^{c-\eps})$-time algorithm for \DyckTwoReach implies
a $O(t(n)^{1-\delta})$ algorithm for problem~$L$ for some $\delta = \delta(\eps)>0$.
This is not unlike usual Turing reductions (allowing multiple queries),
tracking the precise exponents in the running time bounds.
The following result is a consequence of Theorem~\ref{th:cert}.

\begin{theorem}\label{th:nseth-2dyck}
Unless $\NSETH$ fails, there is no fine-grained reduction from $(\SAT, 2^n)$
to $(\DyckTwoReach, n^{\omega+\gamma})$ for any $\gamma > 0$.
\end{theorem}

\section{Certificates for pushdown non-reachability}
\label{s:pda-unreach}

While CFL reachability is a central problem in program analysis,
an analogous problem in model checking is \emph{pushdown reachability} \cite{BEM97,FinkelWW97,BouajjaniEFMRWW00,Schwoon},
formalized as follows.

We are given a pushdown automaton (PDA) $\PDA = (Q, \Gamma, \Delta)$, where
$Q$ is a finite set of states,
$\Gamma$ is a finite alphabet of stack symbols,
and
$\Delta \subseteq
 (Q \times \Gamma) \times
 (Q \times \Gamma^{\le 2})$ is a set of transitions,
and
an initial configuration $(q_0, \gamma_0)\in Q\times \Gamma$. 
We are additionally given a regular set of configurations
$R$ specified by a \emph{$\PDA$-automaton}:
this is a usual, $\eps$-free nondeterministic finite automaton (NFA)
over the alphabet $\Gamma$ in which
the set of control states is $S \supseteq Q$ and the transition
relation is $\delta \sset S \times \Gamma \times S$.
A set of final states, $F \sset S$, is usually taken to be disjoint
from $Q$.
Such a \PDA-automaton is said to accept a configuration $(q, w) \in Q \times \Gamma^*$
of the PDA~\PDA iff there is a walk from control state~$q$ to some $\bar q \in F$
labelled by the word~$w$; in other words, if $w$ is accepted by this
NFA when started from~$q$ as initial state.
We ask if the PDA~$\PDA$ has a run from $(q_0, \gamma_0)$ to some configuration from~$R$.

We adapt our certificate system
to pushdown reachability.
For yes-certificates of size $O(|\Gamma| |S|^2)$,
we can convert the PDA to an equivalent CFG using the standard triplet construction
(see, e.g.,~\cite[Chapter~6]{HopcroftMotwaniUllman})
and repeat the second half of the completeness argument from Subsection~\ref{s:cert:reach}.
Explicitly, a certificate is a ``sub-grammar'' of this CFG that
is a straight-line program.

We now show how to certify that a given initial configuration
\emph{cannot reach} any configuration from a given regular set~$R$.
%
%
The classic saturation algorithm for computing \prestar,
the set of (reflexive, transitive) predecessors of configurations in~$R$, takes
a \PDA-automaton $\mathcal A$ as input and iteratively adds transitions
to it by the following rule:
\begin{equation}
\label{eq:addtrans}
\text{
\PDA has transition $(p, A) \to (q, w)$,
$\mathcal A$ has walk $q \stackrel{w}{\longrightarrow} s$}
\,\Rightarrow
\text{
add transition $p \stackrel{A}{\longrightarrow} s$ to $\mathcal A$.
}
\end{equation}
By the following claim, saturation under~\eqref{eq:addtrans}
implies
overapproximation of \prestar.
The converse inclusion is more subtle and will not be required.
\begin{claim}%
[see, e.g., Carayol and Hague~{\cite[Section~3.2]{CarayolH14}}]
A \PDA-automaton $\mathcal A$ accepts all configurations from~\prestar
if (i) it contains all transitions of the original \PDA-automaton
and (ii) it is saturated, i.e., applying rule~\eqref{eq:addtrans} does not change the transition relation.
\end{claim}
Our certificate system for non-reachability relies on the observation
that the update rule~\eqref{eq:addtrans} can be expressed using
matrix multiplication.
A \emph{certificate} is a finite family of matrices,
$M^A$,
$M^{A,B}$,
$M^{A,B,C}_1$,
$M^{A,B,C}_2$,
for all $A, B, C \in \Gamma$,
satisfying the following conditions:
\begin{equation}
\label{eq:pda-sep}
\begin{aligned}
P^A &\le M^A,
&
&(M^{\gamma_0})_{q_0,f} = 0 & \text{for all\ }&\text{$f \in F$,}
\\
        T^{A,\eps} &\le M^A, \\
\tobool{M^{A,B}} &\le M^A, &
M^{A,B} &= T^{A,B} \cdot M^B, \\
\tobool{M^{A,B,C}_2} &\le M^A, &
M^{A,B,C}_1 & = T^{A,BC} \cdot M^B, &
M^{A,B,C}_2 & = M^{A,B,C}_1 \cdot M^C,
\end{aligned}
\end{equation}
where
we assume with no loss of generality that $S = \{1, \ldots, |S|\}$
and
denote by $P^A$ the $A$-transition matrix of the original \PDA-automaton
and, for all $A \in \Gamma$, $w \in \Gamma^{\le 2}$,
by $T^{A,w} = (t_{i j}^{(A,w)})$ the $0$--$1$ matrix of size $|S| \times |S|$
in which $t_{i j}^{(A,w)} = 1$ if $i, j \in Q$ and \PDA contains
a transition $(i, A) \to (j, w)$.
The following proposition summarises the properties of this system:
\begin{proposition}
\label{p:pushdown-cert}
Certificates have $O(|\Gamma|^3 |S|^2)$ entries.
An instance of PDA emptiness is a no-instance
iff there exists a certificate for it.
The conditions can be verified by
a deterministic algorithm with running time $O(|\Gamma|^3 |S|^\omega)$ or
a randomized algorithm with running time $O(|\Gamma|^3 |S|^2)$ that
accepts valid certificates with probability one and rejects invalid ones with probability $\ge 0.5$.
\end{proposition}

Matrix constraints of Eq.~\eqref{eq:pda-sep}
define a \emph{backwards invariant} for the pushdown system \PDA in question,
an overapproximation of the set of configurations from which $R$ is reachable.

\section{Discussion: Fine-grained landscape and a hardest \DyckTwoReach instance} 
\label{s:fl}

In interprocedural program analysis,
the lack of algorithms
with running time $O(n^{3 - \eps})$
is referred to as ``the cubic bottleneck''.
Heintze and McAllester \cite{HeintzeMcAllester} captured this phenomenon
by the class of ``2NPDA-complete'' problems.
Here ``2NPDA'' stands for \emph{two-way nondeterministic pushdown automata,}
a model of computation that extends standard PDA
with the ability to move back and forth on the (read-only) input tape~\cite{AhoHopcroftUllman68}.
A problem is 2NPDA-complete (following Neal~\cite{Neal})
if it is subcubic equivalent to \emph{2NPDA recognition}:
given a word, does it belong to the language of a fixed 2NPDA.
Heintze and McAllester show a number of 2NPDA-complete problems, including
ground monadic rewriting reachability (see also \cite{Neal}), data flow reachability,
control flow reachability, and certain (non-)typability problems.
Melski and Reps~\cite{MelskiReps} show
a reduction from CFL reachability to data flow reachability and set constraints (and thus to 2NPDA recognition)
and a reverse reduction from data flow reachability to an instance of CFL reachability
where the language is \emph{not} fixed.

The following result appears to be folklore but is not found in the literature,
strengthening the reduction of Melski and Reps to show hardness of CFL
reachability for the \emph{fixed} Dyck-$2$ language.
The equivalence between problems (1) and (2) is sketched by Chaudhuri~\cite{Chaudhuri}.
While we state the result for PDA emptiness, one can equivalently (or additionally) state it for pushdown reachability.
We provide full proofs in the appendix.

\begin{proposition}
\label{prop:eq}
The following problems either all have truly subcubic
algorithms, or none of them do:
(1) 2NPDA language recognition,
(2) PDA language emptiness, and
(3) \DyckTwoReach.
\end{proposition}

\begin{proof}[Proof (sketch)]
We show three reductions:
\begin{itemize}
\item
In 2NPDA recognition to PDA emptiness,
each control state of the PDA remembers the position of the 2NPDA on the input tape
and the control state of the 2NPDA. The size of PDA is linear in the length of the
input word, because the 2NPDA is fixed.
\item
In PDA emptiness to \DyckTwoReach, 
the graph mimics the transition diagram of the PDA.
Stack symbols from $\Gamma$ are encoded by sequences of opening parentheses
of two kinds of length $\lceil \log |\Gamma| \rceil$.
Push transitions are modelled by sequences of edges with these labels,
and pop transitions by sequences with matching closing parentheses.
The reduction is linear-time, because the bit size of the PDA accounts
for the $\log |\Gamma|$ factor.
\item
In the last reduction, we give a fixed 2NPDA that solves \DyckTwoReach. 
The 2NPDA guesses a path through the graph, maintaining at the bottom of the stack
a sequence $\sigma \in \OpenBr^*$, and the current vertex at the top of the stack.
The length of the input word is proportional to the bit size of the graph (adjacency lists).
\qedhere
\end{itemize}
\end{proof}

As a corollary, all of these problems have subcubic certificate schemes,
and an analogue of Theorem~\ref{th:cert} holds for them too
(worked out for PDA emptiness in Section~\ref{s:pda-unreach}).
Theorem~\ref{th:nseth-2dyck} on the absence of \SETH-hardness also
extends to PDA emptiness and 2NPDA recognition.

For upper bounds, note that 2NPDA recognition is solvable in time $O(|w|^3/\log |w|)$~\cite{Rytter85},
and language emptiness for PDA in time $O(n^3/\log n)$\footnote{%
	The reduction of Proposition~\ref{prop:eq}, combined with Chaudhuri's algorithm
        for CFL reachability~\cite{Chaudhuri}, implies
	a $O(n^3/\log n)$ bound for PDA emptiness where $n$ is the bit size of the input.
        (We give a sketch in Appendix~\ref{app:chaudhuri}.)
	In contrast, ``textbook'' algorithms for PDA emptiness go through equivalent context-free grammars~\cite{HopcroftMotwaniUllman}, for which
	a cubic blow-up is unavoidable in the worst case~\cite{GoldstinePW82}.
}.

We observe that the hardness of 2NPDA recognition is witnessed by a single ``hardest'' 2NPDA language:
recognition for an arbitrary 2NPDA can be reduced to a single 2NPDA.
Suppose some 2NPDA \A over $\Sigma$ is given and the input to 2NPDA recognition for $\A$ is a word $w$.
Applying our cycle of reductions from Proposition~\ref{prop:eq} (to PDA emptiness, then to CFL
reachability, and then back to 2NPDA recognition), we get another word
$u = u(\A, w)$ and a 2NPDA $\B = \B(\A, w)$ such that \B accepts $u$ iff \A accepts $w$.
But \B in fact doesn't depend on \A or $w$, because it is
a fixed 2NPDA for \DyckTwoReach.
One refers to such languages as \emph{hardest} 2NPDA languages,
because the recognition problem for $L(\B)$
cannot be easier than the recognition problem for any 2NPDA language $L$.
The following theorem states this result in language-theoretic terms.
(Recall that a homomorphism is a mapping, say $h \colon \Sigma^* \to \Sigma_0^*$,
such that $h(u v) = h(u) h(v)$ for all $u, v \in \Sigma^*$.)

\begin{theorem}
\label{th:hardest}
There exists a 2NPDA $\A_0$ over an input alphabet $\Sigma_0$
with the following property:
for every 2NPDA \A over every finite $\Sigma$ there is a homomorphism $h \colon \Sigma^* \to
\Sigma_0^*$ such that, for all $w \in \Sigma^+$,
$w \in L(\A)$ if and only if $h(w) \in L(\A_0)$.
\end{theorem}

Essentially, $\B = \A_0$.
Working out the details shows that the mapping $u(\A, \cdot)$ can be made a homomorphism
for every \A.
This requires an appropriate encoding for inputs to~$\A_0$.

\begin{remark}
Rytter~\cite{Rytter-hardest} showed there is a fixed hardest 2NPDA language $L_0$,%
\footnote{Actually, Rytter only proves that, for all
    $w \in \Sigma^+$, one has $w \in L$ iff $h(w\$) \in L_0$.
}
based on the classic hardest context-free
language by Greibach~\cite{Greibach-hardest}.
Theorem~\ref{th:hardest} identifies a different hardest 2NPDA language.
In contrast with Rytter's proof, our construction is self-contained
and does not depend on Greibach's hardest CFL.
Instead, our new hardest 2NPDA language is an encoding of a restricted
version of Dyck-2 Reachability.
\end{remark}

We now describe the hardest language $L(\A_0)$.
The alphabet is
$
\Sigma_0 = \{
\opr, \clr, \ops, \cls,
\vmark,
1,
{-},
\edgesep
\}
$.
The language contains only words of the form
\begin{equation}
\label{eq:hardest-input}
\vmark \ell_1 o_1 \edgesep \ell_2 o_2 \edgesep \ldots \edgesep \ell_q o_q \vmark \ell_{q+1} o_{q+1} \ldots \vmark \ldots \ell_m o_m
\end{equation}
and the membership of such words in the language is determined as follows.
Consider a directed graph $G = (V, E)$ with $V = \{1, \ldots, n\}$
where $n$ is the number of \emph{blocks} separated by the vertex marker~\vmark.
An edge $e = (i, j)$ belongs to~$E$ if and only if
the $i$th block has a subword $\ell_p o_p$ with
$\ell_p \in \BrTwo$, $o_p = 1^k$ or $o_p = -1^k$
where $j = i + k$
and this subword is preceded and followed by symbols from $\{\vmark, \edgesep\}$
or tape endmarker.
The edge label is in this case~$\lambda(e) = \ell_i$.
(If for some $i$ and $k$ the index $j$ is ``off the tape'',
 the tape endmarker counts as one virtual vertex and then
 the counting reverses the direction, ``reflecting'' off the endmarker.)
The word belongs to $L(\A_0)$ if and only if
$(G, \lambda, 1, n)$ is a yes-instance of \DyckTwoReach,
i.e., if $G$ contains a walk from $1$ to $n$ labelled by a word from
the Dyck-2 language.

To sum up,
this restricted version of 2NPDA recognition is the ``hard core''
of the problem:
by Theorems~\ref{th:hardest} and~\ref{prop:eq},
in order to find subcubic algorithms for
\DyckTwoReach, it suffices to handle instances obtained from it
(exploiting any structural
properties).
PDA emptiness and \DyckTwoReach are already hard for sparse graphs: 
a truly subcubic algorithm for either problem restricted
to graphs with a linear number of edges
would already result in a breakthrough algorithm for 2NPDA recognition.

\subsubsection*{Acknowledgements.}
We thank Sayan Bhattacharya and Karl Bringmann for interesting discussions.
Philipp Schepper is supported by the European Research Council (ERC) consolidator grant no.~725978 SYSTEMATICGRAPH.
Rupak Majumdar was funded in part by the Deutsche Forschungsgemeinschaft project 389792660-TRR 248 and by the European Research Council under the Grant Agreement 610150 (ERC Synergy Grant ImPACT).

\bibliographystyle{plainurl}
\bibliography{refs}

\newpage
\appendix

\section{Proof of Claim~\ref{claim:pda-reach-to-dyck-2-reach}}
\label{app:pda-reach-to-dyck-2-reach}

Let $\PDA = (Q, \Sigma,\Gamma, \delta, q_0, F)$.
With at most a constant factor blow-up, we can assume that $\PDA$ is in a normal form,
in which each transition 
$(q, a, Z, q', \gamma, d) \in \delta$ is either a ``push'' ($\gamma = Z' Z \in \Gamma^2$)
or a ``pop'' ($\gamma = \epsilon$) or an ``unchanged'' ($\gamma = Z$).

Let $\ell = \lceil \log_2 |\Gamma| \rceil$.
Fix any injective maps $\phi \colon \Gamma \to \OpenBr^\ell$
                   and $\psi \colon \Gamma \to \CloseBr^\ell$
such that, for all $Z \in \Gamma$,
the words $\psi(Z)$ is obtained from $\phi(Z)$
by switching opening brackets to closing brackets
without changing their type---%
i.e., $\opr$ is replaced by $\clr$
  and $\ops$             by $\cls$;
and then reversing the word.

The graph $G'$ has vertices $V \times Q \cup V' \cup \{q_f\}$, where $V'$ is a set of new vertices and $q_f$ is a new vertex.
We shall specify $V'$ later.

The vertices $s' = (s, q_0)$ and $t' = q_f$.
There is a path from $(v, q)$ to $(v',q')$ labeled with the consecutive letters of $\phi(Z')$
if there is an edge $v \xrightarrow{a} v'$ in $G$
and $(q, a, Z, q', Z'Z)\in\delta$.
The intermediate vertices along this path are distinct and are not incident to any other edge of $G'$.
Similarly, there is a path from $(v, q)$ to $(v', q')$ labeled with the consecutive letters of $\psi(Z)$ 
if there is an edge $v \xrightarrow{a} v'$ in $G$ and $(q, a, Z, q', \epsilon)\in\delta$. 
There is an edge $(v, q) \rightarrow (v', q')$ labeled with $\epsilon$ if there is an edge $v \xrightarrow{a} v'$ in $G$
and $(q, a, Z, q', Z)\in\delta$. 
The set of all intermediate vertices added along the way constitute $V'$.
Finally, there is an edge $(t, q) \rightarrow q_f$ labeled with $\epsilon$ for each $q\in F$.

Since $\PDA$ is fixed, the algorithm runs in linear time in $G$ and outputs $G'$ which is linear in the size of $G$.
By induction, we can show that there is a path from $s$ to $t$ in $G$ labeled with a word from $\PDA$ iff
there is a path from $(s, q_0)$ to $q_f$ in $G'$ labeled with a path in Dyck-$2$.

\section{Proof of Lemma~\ref{l:prune-alternatives}}

Let $N$ be the nonterminal from the statement of the lemma.
If $N$ is not productive, i.e., cannot derive any word,
then \emph{all} of its productions can be removed without any effect on
$L(\mathcal G)$. This is simply because $N$ cannot appear in any successful
derivation. We will therefore assume that $N$ is productive.

Consider the parse tree of any successful derivation from~$N$.
We can find in this parse tree a vertex labelled with~$N$
such that none of its descendants is labelled with~$N$.
The subtree $T_N$ rooted at this vertex corresponds to a derivation that
applies some production $P \colon N \to \xi$ first and never
uses $N$ again.

By removing all other productions with left-hand side~$N$
from $\mathcal G$, we obtain a new grammar~$\mathcal G'$.
Let us show that $L(\mathcal G') \ne \emptyset$.
Indeed,
let $S$ be the axiom of~$\mathcal G$.
As $S$ is productive, $u \in L(\mathcal G)$ for some word~$u$.
Consider any parse tree~$T$ of~$u$ in~$\mathcal G$.
If $T$ contains no occurrence of $N$, then it is already a valid parse tree
with respect to~$\mathcal G'$, and we are done.
Otherwise, for every node labelled with~$N$ in~$T$ from which the shortest
path to the root has no other occurrence of~$N$, we replace
the corresponding subtree by $T_N$.
This results in a valid parse tree with respect to~$\mathcal G'$,
because $T_N$ has one occurrence of $N$ only, namely at its root,
where the production applied is $P$.
The new parse tree is a derivation of some word in $L(\mathcal G')$,
which concludes the proof.

\section{Proof of Proposition~\ref{p:pushdown-cert}}

Let $\mathcal A$ be a \PDA-automaton (saturated or not).
For each $A \in \Gamma$,
let $M^A = (m_{i j})$ denote the $A$-transition matrix of $\mathcal A$,
that is, the $0$--$1$ matrix of size $|S| \times |S|$
in which $m_{i j} = 1$ if $\mathcal A$ contains a transition $i \stackrel{A}{\longrightarrow} j$
and $m_{i j} = 0$ otherwise.
Then rule~\eqref{eq:addtrans} can be decomposed into the following
updates, for all $A, B, C \in \Gamma$:
\begin{align*}
M^A &:= \tobool{ M^A + T^{A,\eps} }, \\
M^A &:= \tobool{ M^A + T^{A,B} \cdot M^B }, \\
M^A &:= \tobool{ M^A + T^{A,BC} \cdot M^B \cdot M^C }.
\end{align*}
The composition of certificates~\eqref{eq:pda-sep}
and the existence of verification algorithms follow as
in Subsection~\ref{s:cert:non-reach}.

\section{Fine-grained reductions and proof of Theorem~\ref{th:nseth-2dyck}}
\label{app:fgc}

We discuss further preliminaries on fine-grained complexity,
referrinig the reader to the recent survey by Vassilevska Williams~\cite{WilliamsSurvey}
and to the paper on nondeterministic strong exponential-time hypothesis
by Carmosino et al.~\cite{Carmosino}.

Let $L_1$ and $L_2$ be languages, and let $T_1$ and $T_2$ be time bounds,
i.e., functions $\mathbb N \to \mathbb N$.
We interpret pairs $(L_i, T_i)$ as problems with their conjectured
(or presumed) complexities.
We say that $(L_1, T_1)$ \emph{fine-grained reduces} to $(L_2, T_2)$, written $(L_1, T_1) \leq_{\mathsf{FGR}} (L_2, T_2)$,
if (a) for all $\varepsilon >0$, there is $\delta > 0$  and a deterministic Turing reduction $M^{L_2}$ 
from $L_1$ to $L_2$ such that $\DTIME[M] \leq T_1^{1-\delta}$ and such that
(b) if $Q(M, x)$ denotes the set of queries made by $M$ to the $L_2$ oracle on an input $x$ of length $n$,
then the query lengths obey the time bound
\[
\sum_{q\in Q(M,x)} (T_2(|q|))^{1-\varepsilon} \leq (T_1(n))^{1-\delta}.
\]
Intuitively, a fine-grained reduction from $(L_1,T_1)$ to $(L_2,T_2)$
enables algorithmic savings for $L_2$ to be transferred to $L_1$.
That is, if
$L_2$ can be solved in time~$T_2^{1-\eps}$, then
$L_1$ can be solved in time~$T_1^{1-\delta}$.
A language $L$ with time complexity $T$ is \emph{$\SETH$-hard} if $(\SAT, 2^n) \leq_{\mathsf{FGR}} (L, T)$.

\begin{theorem}[\cite{Carmosino}, Theorem 2 and Corollary 2]
\label{th:carmosino-nseth}
Suppose $\NSETH$ holds and a problem $L$ belongs to $\NTIME[T]\cap \CONTIME[T]$.
Then $(\SAT,2^n)\not \leq_{\mathsf{FGR}} (L,T^{1+\gamma})$ for any $\gamma > 0$.
Also,
for any $L'$ that is $\SETH$-hard with time $T'$, and any $\gamma > 0$, 
we have $(L', T') \not \leq_{\mathsf{FGR}} (L,T^{1+\gamma})$.
\end{theorem}

We are now ready to formulate Theorem~\ref{th:nseth-2dyck} rigorously.

\begin{theorem}[{Theorem~\ref{th:nseth-2dyck} restated}]
\label{th:nseth-2dyck-full}
Unless $\NSETH$ fails, $(\SAT, 2^n) \not \leq_{\mathsf{FGR}}
(\DyckTwoReach, n^{\omega+\gamma})$ for any $\gamma > 0$.
\end{theorem}

It remains to observe that Theorem~\ref{th:nseth-2dyck}
follows from Theorem~\ref{th:cert} and~\ref{th:carmosino-nseth}.

\section{Proof of Proposition~\ref{prop:eq}}
\label{app:equiv}

\paragraph*{Preliminary Definitions}

\emph{Two-way nondeterministic pushdown automata (2NPDA)} \cite{GrayHI67}
are a powerful formalism introduced in 1967 by Gray, Harrison, and Ibarra~\cite{GrayHI67}.
2NPDA have the form
$\twonpda = (Q, \Sigma, \Gamma, \delta, q_0, \br F)$, where 
$Q$ is a finite set of states,
$\Sigma$ are $\Gamma$ are finite alphabets of input and stack symbols, respectively,
$q_0 \in Q$ is the initial state, 
$F\subseteq Q$ is the set of final states,
and a transition relation
$\delta \subseteq
 Q \times \Sigma \times \Gamma \times
 Q \times \Gamma^* \times \set{-1, 0, +1}$.
We assume $\Sigma$ contains two designated ``end of tape'' symbols $\lhd$ and $\rhd$.
We assume that $\Gamma$ contains a designated ``end of stack'' symbol~$Z_0$
such that any transition $(q, \sigma, Z_0, q', w, d)\in \delta$ satisfies $w = Z_0$.
Thus, no transition of $\twonpda$ replaces $Z_0$ on the stack with a different symbol
and no transition pushes $Z_0$.

Informally, the 2NPDA $\twonpda$ has a finite control (states from $Q$)
which reads a symbol of $\Sigma$ on its input tape and the top symbol in $\Gamma$ of a pushdown store. 
Based on the transition relation $\delta$, 2NPDA moves by changing the control state, 
replacing the top symbol of the pushdown store by a finite string of symbols 
(possibly the empty string),
and moving its input head at most one symbol left or right. 
Initially, the 2NPDA is in state $q_0$, and its pushdown store consists of the single symbol $Z_0$. 
The input tape consists of a word $w\in  (\Sigma \setminus \set{\lmarker, \rmarker})^*$
surrounded by a left marker $\lmarker$ and a right marker $\rmarker$ and the 2NDPA scans the left marker~$\lmarker$.

\begin{remark}
We include the endmarkers \lmarker and \rmarker in the set $\Sigma$ here,
even though we did not mention them back in Section~\ref{s:fl}
when specifying the alphabet $\Sigma_0$ for our hardest 2NPDA language.
Naturally, all symbols used by automata (including the endmarkers) should
be included in the tape alphabet of these automata.
\end{remark}

A configuration of the 2NPDA $\twonpda$ is a triple $(q, w\hat{a}x, \gamma)$, where
$q\in Q$, $w, x\in\Sigma^*$, $a\in \Sigma$, and $\gamma\in\Gamma^*$.  
The ``hat'' on $a$ denotes that the machine is currently scanning the letter $a$.
We write $(q_1, a_1 \ldots \hat{a}_i \ldots a_n, Z\gamma) \rightarrow (q_2, a_1\ldots \hat{a}_j \ldots a_n, \gamma'\gamma)$
whenever $(q_1, a_i, Z, q_2, \gamma', d)\in \delta$ for $d \in\set{-1, 0, +1}$, and $j = i+d$.
We require $j\in \set{1,\ldots, n}$, that is, the scan position does not ``fall off'' the input word.
Note that the input tape is not changed, only the scan position may change.
We write $\rightarrow^*$ for the reflexive and transitive closure of $\rightarrow$.
A word $w\in (\Sigma \setminus\set{\lmarker,\rmarker})^*$ is \emph{accepted} by the 2NPDA if
$(q_0, \hat{\lmarker}w\rmarker, Z_0) \rightarrow^* (q, \lmarker w\hat{\rmarker}, Z_0)$ for some $q\in F$.
The language $L(\twonpda)$ of $\twonpda$ is the set of all accepted words (in $(\Sigma\setminus\set{\lmarker, \rmarker})^*$). 

Informally, the 2NPDA has some run that leads it from the initial configuration with the word on the input
tape to a final state.
Wlog, we can assume above that a word is accepted in a final state with the 2NPDA scanning
the right end marker and the pushdown store only contains $Z_0$.
The transition relation is nondeterministic;
we only require that some run is accepting.
For the reader familiar with one-way automata,
we remark that the role of epsilon-transitions is played
by explicit specification of head movements.

A 1NPDA, or just PDA for short, is a 2NPDA such that $\delta \subseteq Q\times \Sigma\times \Gamma \times Q\times \Gamma^*\times \set{0,+1}$.
Informally, the transitions of a PDA do not allow the scan position to move left, so PDA can only move left to right.
PDA accept exactly the context-free languages.
%
In comparison,
2NPDA are surprisingly powerful devices.
In fact, even their deterministic counterparts can recognize
languages such as $\{a^n b^{p(n)} \mid n \ge 0\}$ where $p$ is a fixed polynomial with natural coefficients
and $\{x \hash y \mid x\mbox{ is a subword (factor) of }y\}$%
~\cite{RytterWRAP,Galil77}.

We consider the following decision problems for these machine classes.
The \emph{recognition} problem for a class of machines $\mathcal{C}$ asks, for a fixed machine $M \in \mathcal{C}$ and an input word $w\in \Sigma^*$,
if $w$ is accepted by $M$, i.e., if $w\in L(M)$.
The \emph{emptiness} problem for class $\mathcal{C}$ asks, given a machine $M\in\mathcal{C}$, if $L(M) = \emptyset$.

\paragraph*{Proof.}

Proposition~\ref{prop:eq} follows from
Lemmas~\ref{l:2npda-rec-to-pda-emp},
       \ref{l:pda-emp-to-dyck-2-reach},
   and~\ref{l:dyck-2-reach-to-2npda-rec}, which we prove next.

\begin{lemma}
\label{l:2npda-rec-to-pda-emp}
There exists a linear-time algorithm that,
given a 2NPDA \B and a word $w$,
outputs a PDA \PDA such that:
\begin{itemize}
\item $|\PDA| \le O(|w|)$ for any fixed \B and
\item the language of \PDA is nonempty iff \B accepts $w$.
\end{itemize}
\end{lemma}

\begin{remark*}
In fact, $|\PDA| \le O(|\B| \cdot |w|)$.
\end{remark*}

\begin{proof}
Denote $n = |w|$ and let $S$ be the set of control states of \B.
Construct a PDA \PDA
with the set of control states $Q = \{0, 1, \ldots, n+1\} \times S$.
The first component of the states of \PDA corresponds to a possible position
of the input head of the 2NPDA \B run on $w$.
Indeed, when \B is run on the word $w$,
its head has $n+2$ possible positions: over any
of the $n$ letters of $w$, over the left endmarker,
and over the right endmarker.

PDA \PDA has the initial state $(0, s_0)$, where
$s_0$ is the initial state of \B.
Transitions of the (nondeterministic) PDA \PDA are defined so that \PDA
would simulate the (nondeterministic) computation of \B on $w$.
The stack of \PDA is always the same as the stack of \B,
and the second component of the control state of \PDA the same
as the control state of \B.
Transitions of \B depend on the input letter,
which is available to \PDA, because \PDA `remembers'
in the control state where the input head of \B is positioned---%
and the input word $w$ is fixed.
Transitions of \PDA need not read any letter from the input;
\PDA accepts (rejects) whenever so does \B.
It is straightforward to see that both assertions of the lemma
hold.
\end{proof}

\begin{lemma}
\label{l:pda-emp-to-dyck-2-reach}
There exists a linear-time algorithm that,
given a PDA \PDA,
outputs a directed graph $G = (V, E)$,
labels $\lambda \colon E \to \BrTwo$ and two vertices $s, t \in V$
such that
$(G, \lambda, s, t)$ is a yes-instance of \DyckTwoReach
iff
the language of \PDA is nonempty.
\end{lemma}

\begin{proof}
We show how to construct the required instance of \DyckTwoReach
given a PDA $\PDA = (Q, \Sigma, \Gamma, \delta, q_0, F)$.

The idea is that we encode stack symbols from $\Gamma$
by sequences of words over the alphabet \OpenBr;
pushing symbols on the stack corresponds to
traversing edges of $G$ labeled by opening brackets,
and popping symbols---to traversing edges labeled by closing brackets.

Let $\ell = \lceil \log |\Gamma| \rceil$.
Fix any injective maps $\phi \colon \Gamma \to \OpenBr^\ell$
                   and $\psi \colon \Gamma \to \CloseBr^\ell$
such that, for all $Z \in \Gamma$,
the words $\psi(Z)$ is obtained from $\phi(Z)$
by switching opening brackets to closing brackets
without changing their type---%
i.e., \opr is replaced by \clr
  and \ops             by \cls;
and then reversing the word.

We next construct an auxiliary graph $G' = (V', E')$ with
labels $\lambda \colon E \to \BrTwo$.
The set $V'$ contains $Q$ as a subset.
For each transition
$(q, a, Z, q', \gamma, d) \in \delta$,
the graph $G'$ contains a path from $q$ to $q'$
of length $\ell \cdot (1 + |\gamma|)$.
The edges of this path are labelled by consecutive letters
of the word $\psi(Z) \cdot \phi(\gamma)$;
all intermediate vertices are distinct and are
incident to no other edge of $G'$.
It is easy to see that the number of edges of $G'$
does not exceed $|\PDA|$.
(Notice that the input letter $a$ is ignored in this construction.)

Recall that the automaton \PDA has a nonempty language
if and only if there is a path from its initial configuration
to a final configuration, enabled by some input word from $\Sigma^*$.
The initial configuration $c_0$ of \PDA
has control state $q_0$ and stack content $Z_0$;
and any final configuration $c$ has some control state $q \in F$
and the same stack content $Z_0$.
By construction, $c_0 \rightarrow^* c$ in the PDA \PDA
if and only if the graph $G'$ has a walk from $q_0$ to $q$
labeled by a word $u \in \BrTwo^*$ such that
$\phi(Z_0) \cdot u \cdot \psi(Z_0)$ is a Dyck-2 word.

It now remains to obtain the graph $G$ from $G'$
by adding fresh states $s$ and $t$
and connecting them to the other vertices by
(1) a path from $s$ to $q_0$ labeled by $\phi(Z_0)$ and
(2) paths from each $q \in F$ to $t$ labeled by $\psi(Z_0)$.
Each of these paths has length $\ell$;
paths of type (2) have $\ell-1$ edges in common.
Now $(G, \lambda, s, t)$ is the instance of \DyckTwoReach
with the required property.
\end{proof}

\begin{lemma}
\label{l:dyck-2-reach-to-2npda-rec}
There exist a 2NPDA language \RLang and a linear-time algorithm that,
given a directed graph $G = (V, E)$ with labels $\lambda \colon E \to \BrTwo$
and two vertices $s, t \in V$,
outputs a word $w$ such that
$w \in \RLang$ iff $(G, \lambda, s, t)$ is a yes-instance of \DyckTwoReach.
\end{lemma}

\begin{remark}
For the linear time bound,
we assume that the graph $G$ is encoded in binary in the input.
If this is not the case, the running time of the algorithm
suffers a slowdown by a factor of $O(\log |V|)$.
\end{remark}

\begin{proof}
Words of the language \RLang are encodings of the quadruples $(G, \lambda, s, t)$,
where the vertices of $G$ are encoded in binary.
In more detail, every $w \in \RLang$ has the following form:
first an encoding of $s$,
then an encoding of $t$,
and finally a sequence of encodings of edges of $G$,
where every edge $e \in E$ is followed by its label~$\lambda(e)$.
All these encodings are separated by delimiters.

The language \RLang is over an alphabet of size $O(1)$;
a word belongs to \RLang iff it follows the format we have just described
and the graph $G$ has a walk from $s$ to $t$ labeled
with a sequence from the Dyck-2 language over \BrTwo. 

The algorithm from the assertion of the lemma simply
writes down the encodings in the required format;
it is clear that
the algorithm runs in linear time
and the obtained word belongs to \RLang
iff $(G, \lambda, s, t)$ is a yes-instance of \DyckTwoReach.

It remains to prove that the language \RLang is recognized by a 2NPDA.
Let us describe this 2NPDA~\Rp.
It first reads the input word and checks that it follows the format
described above. If this is not the case, \Rp rejects,
otherwise it guesses the required walk in $G$ from $s$ to $t$
as follows.

A configuration of \Rp stores on the stack the following data:
\begin{itemize}
\item (at the bottom) a sequence $\sigma \in \OpenBr^*$, and
\item (at the top) a vertex $v \in V$.
\end{itemize}
In this configuration, \Rp has already found a walk from $s \in V$
to $v \in V$ labeled with some word $\sigma' \in \BrTwo^*$ that reduces
to $\sigma$. (A word $\sigma'\in\BrTwo^*$ reduces to $\sigma$ if $\sigma$ can
be obtained from $\sigma'$ by a sequence of transformations that replace
the subwords $\opr \clr$ and $\ops \cls$ with $\epsilon$.)

Here is how \Rp works:
\begin{enumerate}
\item
At the beginning, initialize $\sigma$ with the empty word
and $v$ with $s \in V$, pushing them to the stack.
\item
Repeatedly guess the next edge $e \in E$ in the walk
(leaving the loop nondeterministically after some iteration):
\begin{enumerate}
\item
move the head to the encoding of $e = (u_1, u_2)$ written on the input tape;
\item
pop the encoding
of $v \in V$ from the stack, reading the encoding of $u_1$ from
the input tape in sync; if $u_1 \ne v$, reject;
\item
look at the label $\lambda(e)$:
\begin{itemize}
\item if $\lambda(e) \in \OpenBr$, then push $\lambda(e)$ onto the stack,
extending the current $\sigma \in \OpenBr^*$, and
\item if $\lambda(e) \in \CloseBr$, then pop the last symbol of $\sigma \in \OpenBr^*$;
proceed if the two symbols form a matching pair, otherwise
reject (also reject if $\sigma$ is empty);
\end{itemize}
\item
push the encoding of $u_2$ to the stack.
\end{enumerate}
\item
Check if the current vertex $v$ is equal to $t$
and $\sigma$ is empty. Accept if the check succeeds, otherwise reject.
\end{enumerate}
It is easy to see that an accepting computation of \Rp exists
iff $(G, \lambda, s, t)$ is a yes-instance of \DyckTwoReach.
\end{proof}

\section{On an $O(n^3 / \log n)$ algorithm for PDA emptiness}
\label{app:chaudhuri}

As mentioned in Section~\ref{s:fl},
the PDA emptiness to Dyck-2 reachability reduction from Proposition~\ref{prop:eq},
combined with Chaudhuri's algorithm
for CFL reachability~\cite{Chaudhuri}, implies
a slightly subcubic bound for PDA emptiness.

Indeed, Chaudhuri shows how to solve instances of CFL reachability
for a fixed language (including the Dyck-2 language)
in time $O(n^3 / \log n)$, where $n$ is the number of nodes in the graph.

Suppose we start with a PDA emptiness instance with $s$ states, $t$ transitions,
and $r$ stack symbols. Note that we can safely ignore the input alphabet symbols.
The bit size of the instance
is $b = O(t \log (s + r))$.
The reduction from Lemma~\ref{l:pda-emp-to-dyck-2-reach}
gives an instance of Dyck-2 reachability
with $O(b)$ nodes. Chaudhuri solves it in time $O(b^3 / \log b)$,
which is subcubic in the bit size of the input of PDA emptiness
(although not necessarily subcubic in $s+t$).

This complexity seems folklore but was never made explicit. 
In particular, ``textbook'' algorithms for PDA emptiness go through equivalent context-free grammars~\cite{HopcroftMotwaniUllman}, for which
a cubic blow-up is unavoidable in the worst case~\cite{GoldstinePW82}.

\section{Proof of Theorem~\ref{th:hardest}} 
\label{s:hardest:proof}

Fix an arbitrary 2NPDA~\A over a finite alphabet $\Sigma$.
We can assume with no loss of generality that \A has a single
final state and that it is different from its initial state:
$|F| = 1$, $q_0 \not\in F$. (It is an easy exercise to modify
\A to ensure this assumption holds.)

Suppose an input word $w \in \Sigma^*$ is given.
Lemma~\ref{l:2npda-rec-to-pda-emp} reduces
$L(\A)$ to the emptiness problem for a PDA defined as a product
of the word $w$ and 2NPDA $\A$.
More concretely, this PDA has control states $Q = \{0, 1, \ldots, n+1\} \times S$
where $S$ is the set of control states of \A.
Note that $|Q| = O(|w| \cdot |\A|) = O(|w|)$ since \A is fixed.
Here and below,
the constant behind $O(\cdot)$ depends on \A but not on $w$.
Similarly, the stack alphabet of the PDA is fixed too.
We now give this PDA as input to
a further reduction to Dyck-2 Reachability (Lemma~\ref{l:pda-emp-to-dyck-2-reach}),
which produces an instance $(G, \lambda, s, t)$.

\begin{claim}
\label{c:restricted-d2r}
The graph $G$ has the following properties:
\begin{enumerate}
\renewcommand{\theenumi}{(\alph{enumi})}
\renewcommand{\labelenumi}{\theenumi}
\item
\label{cond:count}
    it has $O(|w|)$ vertices (including intermediate ones,
    resulting from mapping the stack alphabet into binary words);
\item
\label{cond:labels}
    its edges are labeled with symbols from \BrTwo;
\item
\label{cond:order-gap}
    there is a linear order on the vertices such that each edge
    connects two vertices that are $O(1)$ positions away from each other
    in this order;
\item
\label{cond:ss-ordering}
    the source is first and the sink is last in the order.
\end{enumerate}
\end{claim}

\begin{claimproof}
Property~\ref{cond:count} is due to the fact that \A, and thus
its stack alphabet, is fixed.
Property~\ref{cond:labels} is immediate.
Property~\ref{cond:order-gap} ultimately reflects the fact that \A,
as a two-way pushdown automaton, cannot jump cells of the input tape,
that is, its head can only move one cell left or right if it moves at all ---%
this is represented by $d \in \{-1, 0, +1\}$ in the syntax of 2NPDA.
Thus, the linear order on vertices of the graph is inherited from
the natural ordering of letters of the input tape, $\lmarker w \rmarker$.
Reductions to PDA emptiness and \DyckTwoReach effectively apply
a direct product construction with a constant factor expansion.
Within each block corresponding to an input letter,
vertices can be ordered arbitrarily, provided that the initial
state of \A comes first and the final state last ---%
ensuring property~\ref{cond:ss-ordering}.
Note that our previous preprocessing of \A
ensures that these two states are different, and
our acceptance condition and subsequent reductions do the rest of the work.
\end{claimproof}

We refer to instances $(G, \lambda, s, t)$ with the properties
stated in Claim~\ref{c:restricted-d2r}
as those of \emph{Restricted Dyck-2 Reachability}.

Suppose $k \in \mathbb N$ is chosen such that
the constants behind $O(\cdot)$ in conditions~\ref{cond:count}
and~\ref{cond:order-gap} are at most~$k$ and every vertex has at most~$k$
outgoing edges. We think of this $k = O(1)$
as the ``width'' of the instance, which depends on the original 2NPDA~\A
but not on~$w$.

\begin{remark}
The constant $O(1)$ in property~\ref{cond:order-gap} is reminiscent to
the bounded pathwidth condition (see, e.g, Bienstock et al.~\cite{BienstockRST91}).
However, in our case
the graph has an even more ``regular'' structure. We leave it open whether
this structure can be characterized by constant pathwidth \emph{and}
constant degree
(and restricting the direction and labels of the edges).
%
%
In comparison, Chatterjee and Osang look at pushdown reachability with constant
 \emph{treewidth}~\cite{ChatterjeeO17}.
\end{remark}

It remains to map this instance of Dyck-2 Reachability to an instance
of 2NPDA recognition, for a fixed 2NPDA which we now define.

For each vertex $v$, let $\ind(v)$ denote the position of~$v$ in the order
specified in property~\ref{cond:order-gap}, ranging from $1$ to $O(|w|)$.
(Once again, the constant behind $O(\cdot)$ depends on \A but not on $w$.)
The construction below follows in spirit the proof of Lemma~\ref{l:dyck-2-reach-to-2npda-rec}
and refines the details in order to produce a homomorphism $h$.
The key difference is that, to produce the new input word,
we will not write edges as ``$(u, v), \lambda(u, v)$''.
Instead we will:
\begin{enumerate}
\renewcommand{\labelenumi}{\theenumi)}
\item
   sort the vertices $u$ according to their $\ind(u)$ ascending
   and, for each~$u$, group all the edges departing from~$u$ together
   (each $u$ will have at most~$k$ outgoing edges);
\item
   write edges $(u, v)$ as pairs $(\lambda(u, v), \offset(u, v))$ where
   $\offset(u, v) = \ind(v)-\ind(u)$,
   i.e., how many vertices to the right the destination of the edge is;
   this difference is written in unary notation (without incurring blowup,
   as this difference cannot exceed~$k$);
\item
   write vertices as ``separators'' between groups of edges.
\end{enumerate}
Putting everything together, the input
to the new 2NPDA has the form~\eqref{eq:hardest-input}
(see page~\pageref{eq:hardest-input}),
where \vmark is the vertex marker symbol, $\ell_i \in \BrTwo$,
and each $o_i$ is either the empty word, or $1\ldots1$ or $-1\ldots1$.

\begin{claim}
\label{c:new-hardest}
The set of valid encodings~\eqref{eq:hardest-input} of Restricted Dyck-2
Reachability can be recognized by a fixed 2NPDA.
\end{claim}

The construction of the 2NPDA in Claim~\ref{c:new-hardest} is similar
to the reduction of Lemma~\ref{l:dyck-2-reach-to-2npda-rec} which we already have.
Instead of guessing the next vertex, this new 2NPDA \AlmostHardestA ``scrolls'' left and
right in a deterministic way to the destination of the current edge, counting
in unary with the help of its stack. The nondeterministic choices that \AlmostHardestA
makes are which outgoing edge from the current vertex to choose next.

Note that the construction of \AlmostHardestA is independent of~$k$,
thus identifying a single hardest language, $L(\AlmostHardestA)$.
Moreover, for a given initial 2NPDA \A this reduction replaces each
symbol in $w$ with $O(1)$ vertices and $O(1)$ edges, where this $O(1)$ depends just
on~\A and not~$w$. The exact collection of these vertices and edges is fully
determined by each symbol of $w$, independently of its position within $w$.
The vertices are not addressed in any ``absolute''
numbering scheme --- so this mapping can be realised as a homomorphism.

\begin{remark}
The use of relative rather than absolute addresses
(to encode $\offset$s)
appears in a related context but for a different problem
in Neal's work on taxonomic inference~\cite{Neal}, which is at the origin
of the connection between 2NPDA and program analysis.
\end{remark}

\paragraph*{Summary and the endmarkers problem.}

We now have achieved the following: for every 2NPDA \A
there is a homomorphism $h_1$ such that $w \in L(\A)$
if and only if
$h_1(\llmarker w \rrmarker) \in L(\AlmostHardestA)$.
Note the appearance of the endmarkers here.
(We use \llmarker instead of \lmarker and
        \rrmarker instead of \rmarker to avoid a notation
 clash in the discussion that follows.)
They reflect the fact that, in the chain of our reductions,
the set of control states of the PDA is
$\{0, 1, \ldots, n+1\} \times S$ not
$\{1, \ldots, n\} \times S$.

To lift our construction from $\llmarker w \rrmarker$ to just $w$,
it may be tempting to
appeal to the following fact, which is not difficult to prove.
Let $x, y \in \Sigma^*$ be fixed.
Suppose a 2NPDA accepts
a language $L \sset x \cdot \Sigma^* \cdot y$.
Then there exists another 2NPDA which accepts
the language $\{ w \mid x w y \in L \}$.

Unfortunately, this fact does not quite achieve our goal.
This is because the new 2NPDA we would obtain from it
depends on $x$ and $y$. In our context, $x$ and $y$ should
be the images of the original endmarkers, i.e.,
we would like to have $x = h_1(\llmarker)$ and $y = h_1(\rrmarker)$.
But these two words depend on the homomorphism~$h_1$, and thus
on the 2NPDA~\A that we started from.
This is at odds with our objective: we need a single 2NPDA
for our hardest language, not an entire family dependent on \A.

There are several ways to deal with this issue.
One is reminiscent of Rytter's approach~\cite{Rytter-hardest}:
we can decide we are content with keeping a single endmarker in,
i.e., we would only like
to find an $L_0$ such that, for all
$w \in \Sigma^+$, one has $w \in L$ iff $h(w\$) \in L_0$.
Here $\$$ is a fresh symbol. It is not very difficult
to find such an $h$ and $L_0$ based on our construction:
essentially, the word $h_1(\llmarker)$ needs to be merged
with the word $h_1(\rrmarker)$ and placed to the right of $h_1(w)$.
So we would like to choose
$h(\$) = h_1(\rrmarker) h_1(\llmarker)$ and
$h(a) = h_1(a)$ for all other symbols~$a$.
The 2NPDA for $L_0$ is the same as our 2NPDA~\AlmostHardestA constructed
above, with the following modification.
Suppose it starts following an edge from some vertex (block) to the left
but hits the left end of the tape, i.e., the left endmarker~\lmarker. We now
use this symbol to refer to the tape alphabet of the 2NPDA~\AlmostHardestA
(and not the tape alphabet of the original machine~\A).
The new 2NPDA will move all the way to the right end of the tape
and continue its search for the destination vertex from the
right endmarker~\rmarker.
Edges within $h_1(\rrmarker)$ need not be changed, but the ones
among them that lead to the right ($\offset(u,v) > 0$, or equivalently
$o_i \in 1^+$) will make the 2NPDA hit~\rmarker, go back to the left
of the tape and continue the seart from~\lmarker.

One further technicality that needs to be dealt with is the beginning
and end of the computation. Recall that our Restricted Dyck-2 Reachability
asked for a path from the very first vertex to the very last one.
Since $h_1(\llmarker)$ moved, we now need to change this convention.
More concretely,
the only two vertices that we can distinguish correspond to the
\emph{last two} control states and the head position over the
\emph{left} endmarker.
So the original 2NPDA~\A needs to be changed accordingly.

While this approach recovers Rytter's result,
we show below that there is a way to eliminate the extra symbol~$\$$
altogether.

\paragraph*{Merging endmarker blocks into other symbols.}

Our solution to the endmarkers problem acknowledges
that the words $h_1(\llmarker)$ and $h_1(\rrmarker)$ cannot be
eliminated completely. Indeed, the vertices and edges that these
two words encode correspond to the behaviour of the original 2NPDA~\A
over the tape endmarkers, and this behaviour can contribute to the
computations of~\A in a nontrivial way.

However, what we can do is to embed all this information
into words $h(a)$ for \emph{all} other symbols~$a$.
For a first intuition (to be amended later), we would like to set
$h(a) = h_1(\llmarker) \shuf h_1(a) \shuf h_1(\rrmarker)$
for all non-endmarker symbols~$a$,
where $\shuf$ denotes a specially tailored
ternary version of the perfect shuffle operation.
More concretely, let $w_1, w_2, w_3$ be arbitrary words
such that, for some single~$\ell$,
we have $w_i = \prod_{j=1}^{\ell} \vmark w_{i,j}$
where none of the words $w_{i,j}$ contains the vertex marker symbol~\vmark.
Then
\begin{equation*}
w_1 \shuf w_2 \shuf w_3
:=
\prod_{j=1}^{\ell}
\vmark w_{1,j}
\vmark w_{2,j}
\vmark w_{3,j}
\enspace.
\end{equation*}
Note that this shuffling relies on $\ell$ being the same for all
three arguments, and ultimately this means the same number
of vertices (blocks) in all words~$h_1(a)$.
This is in fact ensured by our constructions above
(although we could always achieve this by introducing extra
 dummy vertices where necessary).

As a result of
this shuffling arrangement, we can think of new input
words as having three interleaving ``tracks'', each containing a separate
sequence of vertices.
Naturally, this requires some changes to the wiring, as follows.

First, the offsets that specify the edges of the graph departing
from the vertices of $h_1(a)$ need to be updated.
This is not difficult.
Recall that edge destinations are specified using relative addresses of
vertices.
For every edge from a vertex in $h_1(a)$,
its offset needs to be multiplied by~$3$, so that the edge skips
intermediate vertices from copies of $h_1(\llmarker)$ and
$h_1(\rrmarker)$.

Second, we need to provide a way for the new 2NPDA~\HardestA
to reach the vertices in $h_1(\llmarker)$ and $h_1(\rrmarker)$.
To achieve this, we consider the scenario in which \HardestA~will
traverse edges leading to a vertex in $h_1(\llmarker)$.
(The case of $h_1(\rrmarker)$ is handled in a symmetric way.)
Suppose the head of the 2NPDA is over a block (vertex)
within the leftmost $h_1(a)$. Taking an edge with a negative
offset, it moves left but then hits the left tape endmarker~\lmarker.
When it does so, the stack of the 2NPDA still contains the number
of vertices to be skipped.
The 2NPDA then needs to change from the second (main) track,
which contains the information from $h_1(a)$s, to the first track,
which stores multiple copies of the word~$h_1(\llmarker)$.
Effectively, this amounts to treating~\lmarker as just another
\vmark that on top of its usual function makes the machine
change direction.
After that, however, we see that
the number of vertices to be skipped was counted from the right
of $h_1(\llmarker)$ and not from the left where the head of
the automaton is now located.
Thus, we \emph{reverse} the encoding of each of $h_1(\llmarker)$
and $h_1(\rrmarker)$, as follows:
we re-define our special shuffle as
\begin{equation*}
w_1 \shuf w_2 \shuf w_3
:=
\prod_{j=1}^{\ell}
\vmark \negate w_{1,\ell+1-j}
\vmark w_{2,j}
\vmark \negate w_{3,\ell+1-j}
\enspace,
\end{equation*}
where $\negate u$ is the same word as $u$ in which every maximal
subword of the form $-1^m$ is replaced with $1^m$ and
each $1^m$, without a preceding $-$, with $-1^m$.
Our 2NPDA must remember, in its control state, which ``track'' of the input
it is over. The second track corresponds to the usual operation.
Over the first track:
\begin{itemize}
\item Edges previously specified by positive offset needs to
followed to the \emph{left} instead of to the right (hence the
$\negate w_{1,\ell+1-j}$ above and not $w_{1,\ell+1-j}$. If
the left tape endmarker~\lmarker is encountered, the automaton
transitions to the second (main) track,
and only then continues to the right (in the normal mode).
\item Edges specified by the negative offset need to be
followed to the \emph{right} instead of to the left
(again, this matches the $\negate w_{1,\ell+1-j}$ above).
We note that if the input word for our 2NPDA is the homomorphic
image under~$h$ of some word in~$\Sigma^*$, then the automaton
will never leave the leftmost $h(a)$ while being on the first track,
because the original
2NPDA~\A cannot move left from the left endmarker.
\end{itemize}
The third track is arranged in a similar way.

Importantly, while we apply these changes to $h$ and the ``wiring'' of the
graph, we can keep the semantics of our hardest language untouched.
The ``tracks'' themselves need not enter the description of the language.
The only new ``feature'' that is necessary is \emph{changing} the tracks ---
and this can be achieved simply by specifying that
when our new 2NPDA~\HardestA encounters a tape endmarker during
its operation, this endmarker is counted as a virtual vertex
and ``reflects'' off it, continuing the countdown in the opposite direction.

However, as was the case with the approach described above and involving $\$$,
our new construction of~$h$ breaks the convention about the
source and target vertices in the Dyck-2 reachability instance
(albeit in a slightly different way).
Because of the effective reversal of vertex ordering within
$h_1(\llmarker)$ and $h_1(\rrmarker)$,
the required adjustment to the original 2NPDA~\A is
that its initial control state needs to be the \emph{last}
and its (only) final control state the \emph{first}
in the ordering.

To sum up, by applying these adjustments and ``compiling''
the homomorphism~$h$ the way we have described,
we arrive at the desired 2NPDA~\HardestA.
(Note that there is freedom in whether we take
$\eps \in L(\HardestA)$ or $\eps \not\in L(\HardestA)$,
but as some 2NPDA languages contain $\eps$ and some do not,
their homomorphic images will necessarily disagree
on $\eps$, no matter our choice of the homomorphism.)
This completes the proof.

\end{document}